\documentclass[letterpaper, 10 pt, conference]{ieeeconf}  

\IEEEoverridecommandlockouts                              
\usepackage{amsmath} 
\usepackage{amssymb}  
\usepackage{moreverb,url}
\usepackage[colorlinks,bookmarksopen,bookmarksnumbered,citecolor=red,urlcolor=red]{hyperref}
\usepackage{caption}
\usepackage{subcaption}
\usepackage{url}
\usepackage{amssymb}
\usepackage{latexsym}
\usepackage{psfrag}
\usepackage{epsfig}
\usepackage{graphicx}
\usepackage{epstopdf}
\usepackage{color}
\usepackage{tikz}
\usepackage{upquote}
\usetikzlibrary{calc}

\newcommand\BibTeX{{\rmfamily B\kern-.05em \textsc{i\kern-.025em b}\kern-.08em
T\kern-.1667em\lower.7ex\hbox{E}\kern-.125emX}}

\newcommand{\rank}{{\rm rank\;}}
\newcommand{\R}{\mathbb{R}}
\newcommand{\dfb}{\stackrel{\Delta}{=}}

\newtheorem{theorem}{\textbf{Theorem}}[section]

\newtheorem{definition}[theorem]{\textbf{Definition}}

\title{\LARGE \bf
On the stability and applications of distance-based flexible formations
}

\author{Hector Garcia de Marina, Zhiyong Sun and Shaoshuai Mou
\thanks{Hector Garcia de Marina is with the Unmanned Aerial Systems Center, M{\ae}rsk Mc-Kinney M{\o}ller Institute, Southern University of Denmark, Denmark. (e-mail: {\tt\small hgm@mmmi.sdu.dk}). 
Zhiyong Sun was with  
Research School of Engineering, The Australian National University, Canberra. He is now with the Department of Automatic Control, Lund University, Sweden. (e-mail: {\tt\small sun.zhiyong.cn@gmail.com}) Shaoshuai Mou is with School of Aeronautics and Astronautics, Purdue University, USA. (e-mail: {\tt\small mous@purdue.edu}).
}
}%

\begin{document}

\maketitle
\thispagestyle{empty}
\pagestyle{empty}

\begin{abstract}
This paper investigates the stability of distance-based \textit{flexible} undirected formations in the plane. Without rigidity, there exists a set of connected shapes for given distance constraints, which is called the ambit. We show that a flexible formation can lose its flexibility, or equivalently may reduce the degrees of freedom of its ambit, if a small disturbance is introduced in the range sensor of the agents. The stability of the disturbed equilibrium can be characterized by analyzing the eigenvalues of the linearized augmented error system. Unlike  infinitesimally rigid formations, the disturbed desired equilibrium can be turned unstable regardless of how small the disturbance is. We finally present two examples of how to exploit these disturbances as design parameters. The first example shows how to combine rigid and flexible formations such that some of the agents can move freely in the desired and locally stable ambit. The second example shows how to achieve a specific shape with fewer edges than the necessary for the standard controller in rigid formations.
\end{abstract}

\section{Introduction}
The past decade has seen increasingly rapid advances in the field of formation control. Formation control algorithms form part of the collective intelligence in the deployment of multi-agent systems, which are relevant in the exploration and surveillance of areas among other tasks \cite{oh2015survey}. In particular, rigid formations based on either distance rigidity \cite{KrBrFr08} or bearing rigidity \cite{zhao2014bearing} have emerged as powerful tools for the realization and stabilization of geometrical shapes by a team of agents. With a lot progress achieved in controlling rigid formations, especially using the well-known gradient-descent method \cite{KrBrFr08}, little attentions has been paid to formations without rigidity (i.e. flexible formations), for which the gradient-
descent control is not applicable \cite{dimarogonas2008stability}. We note that flexible setups have relevance in mechanical designs such as leg mechanisms \cite{nansai2013dynamic}; an example is the Jansen's linkage shown in Figure \ref{fig: jansen}, where the rotation of some flexible links induces desired trajectories to the rest of the junctions. This application can be of utility in multi-agent systems, where by controlling a small set of the agents, we can induce non-trivial trajectories to the rest of the agents. 

We define \emph{the ambit} of a flexible formation as all the possible deformations that a flexible shape can adopt for given distance constraints. Unlike  distance-based rigid formations \cite{MouMorseBelSunAnd15}, research to date has not shown the impact of small disturbances in the agents' range sensors on the stability of the desired ambit of a flexible formation. Recognition of this has motivated us to investigate stability of flexible formations under such small disturbances. In particular, we will show that unlike rigid formations, the perturbed desired ambit can be turned unstable regardless of how small the disturbance is, which, therefore, may compromise a multi-agent system if for example the communication ranges between neighbors are critical. We will present a technique based on adding virtual edges to the flexible formation until it becomes virtually rigid. Then, an eigenvalue analysis of the resulting \emph{augmented} error system can characterize the stability of the perturbed flexible formation. Furthermore, we will show that in general, these disturbances will result in losing part or all of the flexibility of the flexible formation. We will exploit this effect in order to control specific shapes like it is done with rigid formations, but requiring fewer distances to be controlled.

The paper is organized as follows. We review the distance-based formations but focus on the flexible ones in Section~\ref{sec: distance}. We continue in Section \ref{sec: error} by suggesting an eigenvalue analysis of an augmented error system in order to characterize the perturbed equilibrium by small disturbances in the range sensors of the agents. Preliminary examples of how to exploit these disturbances in flexible formations are shown in Section \ref{sec: ex}. We finish the paper with some conclusions in Section \ref{sec: con}.

\begin{figure}
\centering
	\includegraphics[width=0.24\columnwidth]{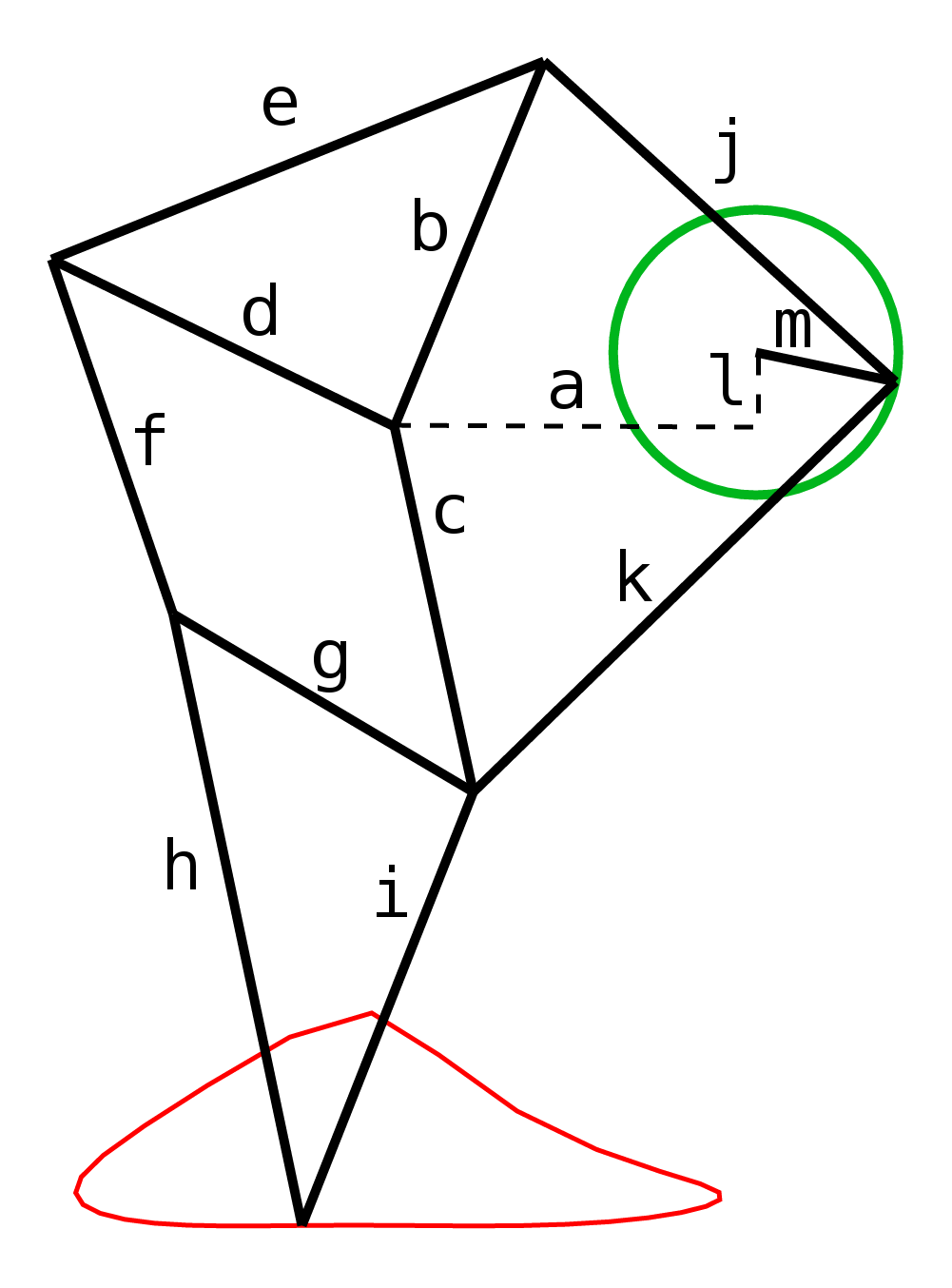}
	\includegraphics[width=0.24\columnwidth]{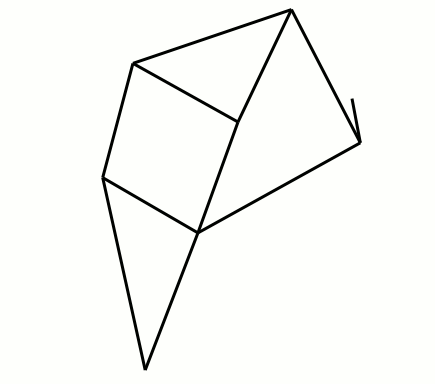}
	\includegraphics[width=0.24\columnwidth]{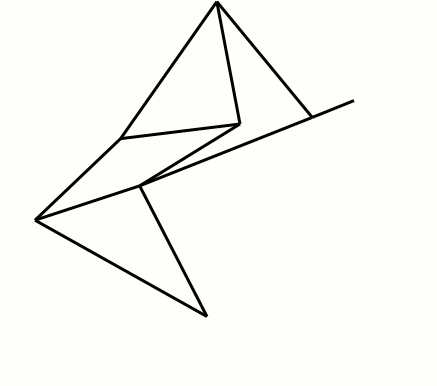}
	\includegraphics[width=0.24\columnwidth]{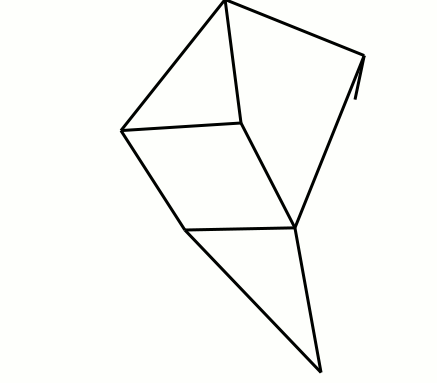}
	\caption{The Jansen's linkage\textsuperscript{1} is a planar leg mechanism to simulate a smooth walking motion. It is based on a flexible framework, where only some distances between nodes or joints are fixed. In particular, one of the sides of the link \emph{m} is fixed, and together with a rotational motion (in green color) actions the leg mechanism. This rotational motion of the edge \emph{m} induces the trajectory described by the node at the bottom in red color. Starting from the left as the initial state, the following three pictures illustrate the evolutions of the shape as the edge \emph{m} rotates clockwise. Pictures courtesy of Michael Ferry, Wikimedia Commons.}
\label{fig: jansen}
\end{figure}
\par\setcounter{footnote}{1}
\footnotetext{For an animation of the Jansen's linkage we refer to the Wikipedia entry \url{https://en.wikipedia.org/wiki/Jansen's\_linkage}}
    
\section{Distance-based formations review}
\label{sec: distance}
\subsection{Rigid, flexible formations and their ambit}
Consider a number of $n$ agents in the plane, whose labels are in the set $\mathcal{N}=\{1,2,...,n\}$. We denote by $p_i \in\mathbb{R}^2$ the position of agent $i\in\mathcal{N}$, and by $p\in\mathbb{R}^{2n}$ the stacked vector of positions $p_i, \forall i\in\mathcal{N}$. An agent $i$ can measure its relative position from other agents in the subset $\mathcal{N}_i \subseteq \mathcal{N}$, i.e., the neighbors of agent $i$. The neighbor relationships are assumed to be symmetric and can be described by an undirected graph $\mathbb{G} = (\mathcal{N}, \mathcal{E})$ with the ordered edge set $\mathcal{E}\subseteq\mathcal{V}\times\mathcal{V}$ for which an edge $(i,j)\in \mathcal{E}$ if and only if $i$ and $j$ are neighbors. The set $\mathcal{N}_i$ is defined by $\mathcal{N}_i\dfb\{j\in\mathcal{N}:(i,j)\in\mathcal{E}\}$. By assigning an arbitrary direction to each undirected edge, one could define the incidence matrix $B=[b_{ik}]_{n\times |\mathcal{E}|}$ for $\mathbb{G}$ by
\begin{equation}
	b_{ik} \dfb \begin{cases}+1 \quad \text{if} \quad i = {\mathcal{E}_k^{\text{tail}}} \\
		-1 \quad \text{if} \quad i = {\mathcal{E}_k^{\text{head}}} \\
		0 \quad \text{otherwise}
	\end{cases},
	\label{eq: B}
\end{equation}
where $\mathcal{E}_k^{\text{tail}}$ and $\mathcal{E}_k^{\text{head}}$ denote the tail and head nodes, respectively, of the edge $\mathcal{E}_k$, i.e., $\mathcal{E}_k = (\mathcal{E}_k^{\text{tail}},\mathcal{E}_k^{\text{head}})$.

The stacked vector of relative positions between neighboring agents is then given by
\begin{equation}
z = \overline B^T p,
	\label{eq: z}
\end{equation}
where $\overline B \dfb B \otimes I_2$ with $I_2$ being the $2\times 2$ identity matrix, and $\otimes$ the Kronecker product. Note that each $z_k \in \mathbb{R}^2$ in $z$ corresponds to the relative position between the neighboring agents $i$ and $j$ in the edge $\mathcal{E}_k = (i,j)$. We denote a \emph{formation} by the pair $\{\mathbb{G},p\}$. Consider the \emph{edge function} $\phi: \R^{2|\mathcal{E}|} \to \R^{|\mathcal{E}|}$ defined by $\phi(\overline B^Tp) \dfb D_z^Tz$, where we define the operator $D_x \dfb \operatorname{diag}\{x_1, \dots, x_l\}$ and $l$ is the number of stacked block elements in $x$. The distance-based formation control in the literature mainly focuses on \emph{infinitesimally rigid} formations, which can be characterized by introducing the \emph{rigidity matrix} $R \dfb D_z^T\overline B^T = \frac{1}{2}\frac{\partial \phi}{\partial p}$. More precisely, a formation is said to be infinitesimally rigid in the plane if $\rank(R) = 2n-3$. Roughly speaking, if one wants to keep the edge function constant by continuously moving the agents of an infinitesimally rigid formation, then the only allowed motion is the combination of translations and rotations of the whole team. In addition, the condition $\rank(R) = 2n-3$ implies a \emph{generic} deployment of the agents on the plane, e.g., $z\neq 0$ and not having all the relative positions $z_k$ in the same line. If the rank condition is not satisfied, then \emph{infinitesimally flexible} motions are allowed. Therefore the formation is not infinitesimally rigid even if it is still rigid \cite{izmestiev2009infinitesimal}.

An infinitesimally rigid formation is also \emph{minimal} if the removal of any one edge in $\mathbb{G}$ causes the formation to lose its rigidity. When a formation loses its rigidity, then it becomes a flexible formation.
\begin{definition}
A formation is \emph{flexible} if the agents can be moved continuously while the distances between neighboring agents remain unchanged, i.e., $\phi(\overline B^Tp)$ remains constant, and at least one inter-agent distance between two non-neighbors changes.
\label{def: flex}
\end{definition}

Consider the rotation matrix $T_k\in SO(2)$, then the set of all rotated vectors $T_kz_k$ such that $\phi$ is constant is defined as the transformation group $\mathcal{T}$. Note that by doing so, these rotations in the plane have their center of rotation at $p_j$ for each $z_k$ such that $\mathcal{E}_k = (i,j)$. Indeed, these rotations do not need to be all equal. For example, if we start from applying an arbitrary rotation to the edge $m$ in Figure \ref{fig: jansen}, then sequentially we can apply different rotations to the rest of the edges for completing one transformation in $\mathcal{T}$. However, while we could have some freedom in choosing the rotation for the next edge, we might be constrained by the previous chosen rotations in order to keep $\phi$ constant. Note that the only transformation $\mathcal{T}_{inf}\in\mathcal{T}$ allowed to an infinitesimally rigid formation is the trivial $Tz_k, \forall k$ for some fixed rotation matrix $T\in SO(2)$, i.e., a rotation of the whole team of agents at once. The term \emph{ambit} will be used in this paper to refer to all the possible \emph{deformations} that a formation can have.
\begin{definition}
	The \emph{ambit} of the formation, denoted by $\mathcal{A}z$, is the set $\{\gamma(z) : \gamma\in\mathcal{T} \}$ where $z$ is as in (\ref{eq: z}), and $\mathcal{T}$ is the group of rotational transformations to the elements of $z$ such that $\phi(\overline B^Tp)$ is constant. In addition, for any composition $\gamma = \gamma_1 \circ \gamma_2$, we exclude the whole rotation of the formation, i.e., $\gamma_{\{1,2\}}\notin \mathcal{T}_{\inf}$, where $\mathcal{T}_{\inf} \dfb Tz_k, k\in\{1,\dots,|\mathcal{E}|\}$ with $T\in SO(2)$.
\label{def: ambit}
\end{definition}

For example, all the possible deformations of the mechanism in Figure \ref{fig: jansen} are in the ambit of the corresponding flexible formation, and note that a whole rotation of one of the \emph{deformed} shapes will not actually \emph{change} such a shape. Therefore the ambit of an infinitesimally rigid formation is empty, since its shape must be constant. The complementary concept of \emph{orbit} is defined in \cite{MouMorseBelSunAnd15}, where the authors focus on infinitesimally rigid formations. Roughly speaking, an orbit refers to the shape that the team of agents form up to translations and rotations in the plane. Note that the ambit of the formation focuses on the relative positions in $z$ and not on the absolute ones in $p$, e.g., how all the possible flexible shapes look like and not where they are in the plane. Furthermore, all the shapes in the ambit of a formation can be translated and rotated on the plane, therefore describing an orbit as well. Since the whole rotation of the shape is excluded of the ambit by definition, then both ambit and orbit do not overlap.

In one of the applications in this paper, we will show that an algorithm will make the agents to converge to a desired ambit where the agents will be moving, while the ambit converges to a specific point in its orbit. For example, the agents would converge to the motion depicted in Figure \ref{fig: jansen}, but this motion will not be translated or rotated in the plane.

\subsection{Gradient descent control}
\label{sec: grad}
We model the dynamics of the agents as simple kinematic points
\begin{equation}
\dot p = u,
	\label{eq: u}
\end{equation}
where $u\in\R^{2n}$ is the stacked vector of control actions $u_i\in\R^2, \forall i$. For an application of the results derived in this paper to second order dynamics, we refer to the techniques shown in \cite{de2018taming}.

We denote by $d_k > 0$ the target distance to be controlled by the two neighboring agents $i$ and $j$ in the edge $\mathcal{E}_k$, and we assume that these distances are feasible, so it is possible to construct the vector of relative positions $z$. Consequently, we define the following error distance
\begin{equation}
	e_k(z) = ||z_k||^2 - d_k^2,
\end{equation}
and by taking the gradient descent of the quadratic potential function 
\begin{equation}
	V = \sum_{k=1}^{|\mathcal{E}|} e_k^2,
	\label{eq: Ve}
\end{equation}
we arrive at the control law proposed in \cite{KrBrFr08}, that together with (\ref{eq: u}) yields the compact form
\begin{equation}
	\dot p = -R(z)^Te,
	\label{eq: pdot}
\end{equation}
where $R(z)$ is the rigidity matrix, and $e\in\R^{|\mathcal{E}|}$ is the stacked vector of $e_k$'s. For the sake of completeness, the dynamics of agent $i$ extracted from (\ref{eq: pdot}) are given by
\begin{equation}
	\dot p_i = -\sum_{j\in\mathcal{N}_i}(p_i - p_j)(||p_i - p_j||^2 - d_{ij}^2),
	\label{eq: dotpigrad}
\end{equation}
where $d_{ij}=d_{ji}$ corresponds to $d_k$ for the edge $\mathcal{E}_k = (i,j)$.

The following systems regarding the dynamics of the relative positions and the distance errors will be used throughout the paper
\begin{align}
	\dot z &= -\overline B^TR(z)^Te \label{eq: zdot} \\
	\dot e &= -2D_z^T\overline B^TR(z)^Te = -2R(z)R(z)^Te, \label{eq: edot}
\end{align}
where we have applied (\ref{eq: z}) and $\dot e_k = 2z^T_k\dot z_k$ for deriving (\ref{eq: zdot}) and (\ref{eq: edot}) respectively.
The local stabilization of a minimally infinitesimally rigid shape can be shown by choosing (\ref{eq: Ve}) as a Lyapunov function, whose time derivative is
\begin{equation}
	\dot V = -4e^TR(z)R(z)^Te = -4e^TQ(e)e, \label{eq: dotVe}
\end{equation}
where all the terms in $R(z)R(z)^T = Q(e)$ can be expressed as functions of $e$ if the formation is infinitesimally rigid \cite{MouMorseBelSunAnd15}, and the matrix $Q(0)$ is positive definite since $R^T(z)$ is full row rank for minimally infinitesimally rigid formations. Consequently, the error system (\ref{eq: edot}) is self-contained and the local (exponential) stability to the desired shape and to a point in its orbit on the plane follows.

\section{Stability analysis of flexible formations}
\label{sec: error}

The goal of this section is to check the robustness of the desired equilibrium of the flexible formation under small disturbances in the range sensors of the agents, e.g., a small bias. Interestingly, we will show that there is an important difference between infinitesimally rigid and flexible formations regarding their stability.

For a desired infinitesimally rigid formation, different shapes can be constructed from the set
\begin{equation}
	\mathcal{D} \dfb \big\{z : ||z_k|| = d_k, k\in\{1,\dots,|\mathcal{E}|\}\big \},
\label{eq: D}
\end{equation}
likewise different ambits can be constructed from the same set if the formation is flexible. Therefore, the action of the controller for a flexible formation must stabilize the team of agents in a desired ambit, possibly making the team to converge to a generic point in the ambit like in \cite{dimarogonas2008stability}, or in the  stronger version, to a desired point in the ambit likewise the standard approach controls rigid shapes \cite{KrBrFr08}.

Note that according to the definition of rigidity, the union of a rigid and a flexible formation will lead to another flexible formation. We consider that the graph $\mathbb{G}_{flex}$ in a flexible formation does not contain any rigid subgraph, and that the graph $\mathbb{G}_{inf}$ refers to a minimally and infinitesimally rigid formation. For the sake of simplicity, we will first analyze the equilibrium points of a formation where $\mathbb{G} = \mathbb{G}_{flex}$, and together with the results on the stability of rigid formations \cite{MouMorseBelSunAnd15} we can reach some conclusions for the composition of rigid and flexible formations.

\subsection{The augmented error system}
\label{sec: aug}
In this section we will build a self-contained error system for flexible formations. Recall that the number of edges in $\mathbb{G}_{flex}$ is fewer than in $\mathbb{G}_{inf}$, since $\mathbb{G}_{flex}$ does not contain any subgraph that can be part of an infinitesimally rigid formation. Let us \emph{virtually} add the minimum number of edges to the graph $\mathbb{G}_{flex}$ until it becomes $\mathbb{G}_{inf}$. That is, roughly speaking, the formation will virtually become a minimally infinitesimally rigid one and consequently we will be able to construct a self-contained error system \cite{MouMorseBelSunAnd15}. We employ the word \emph{virtual} since the new edges do not have any impact in the actual controller in (\ref{eq: pdot}), but to make the analysis of the error system (\ref{eq: edot}) easier. We denote this new graph by $\mathbb{\tilde G}_{inf}$, and define by $\tilde B$ the incidence matrix only for these virtual edges. Thus, the new relative positions for these virtual neighbors are given by $\tilde z = \overline{\tilde B}^Tp$. Without loss of generality one picks a point in the ambit of the desired flexible formation. In particular, an arbitrary shape can be infinitesimally rigid because of the virtual $\tilde z$. For example, all agents must not be collinear or coincident in the same position. Then, assign the corresponding distances $\tilde d_k$ to the new created edges for constructing their virtual error signals. Note that with these virtual operations we leave untouched the controller in (\ref{eq: pdot}), i.e., these new distances $\tilde d_k$ are not controlled at all. In particular, we have a new virtual error vector $\tilde e\in\R^{|\mathcal{E}_{inf} | - |\mathcal{E}_{flex}|}$, and note that in the dynamics
\begin{equation}
	\dot{\tilde z} = \overline{\tilde B}^T \dot p = -\overline{\tilde B}^T\overline BD_z e,
\end{equation}
the signal $\tilde e$ is not involved. Also note that since the origin of the error $e$ is locally stable as we discussed in Section \ref{sec: grad}, then the $\tilde z$ dynamics are also stable. We now derive the dynamics for the new virtual error signal
\begin{equation}
	\dot{\tilde e} = -2D^T_{\tilde z}\overline{\tilde B}^T\overline BD_z e.
	\label{eq: etildedot}
\end{equation}
Since $\mathbb{\tilde G}_{inf}$ comes from a minimally infinitesimally rigid formation, the dynamics of the stacked vector $\begin{bmatrix}e^T & \tilde e^T\end{bmatrix}^T$, or the \emph{augmented error}, is now self-contained. In particular, all the dot products involving the elements of $z$ and $\tilde z$ can be written as functions of $\begin{bmatrix}e^T & \tilde e^T\end{bmatrix}^T$. Now we are ready for the following result addressing the robustness of the desired flexible formation.

\begin{theorem}
	\label{thm: 1}
	Consider a flexible formation $\{\mathbb{G}_{flex},p\}$ with a set of feasible desired distances and a controller as in (\ref{eq: pdot}). Then, for all virtual $\{\mathbb{\tilde G}_{inf},p\}$ constructed from $\{\mathbb{G}_{flex},p\}$, the linearization at the origin of the self-contained augmented error $\begin{bmatrix}e^T & \tilde e^T\end{bmatrix}^T$ dynamics contains at least one zero eigenvalue.
\end{theorem}
\begin{proof}
We first write the dynamics of $\begin{bmatrix}e^T & \tilde e^T\end{bmatrix}^T$ together
\begin{equation}
\begin{cases}
	\dot e &= -2R(z)R(z)^Te = -Q(e, \tilde e) e =: \epsilon_1(e,\tilde e)\\
	\dot{\tilde e} &= -2D^T_{\tilde z}\overline{\tilde B}^TR(z)^T e = -F(e, \tilde e)e =: \epsilon_2(e,\tilde e)
\end{cases},
\end{equation}
where the functions $Q$ and $F$ only depend on $e$ and $\tilde e$ since $\{\mathbb{\tilde G}_{inf},p\}$ is infinitesimally and minimally rigid. In particular, these functions actually depend on the chosen $d_k$ and $\tilde d_k$, i.e., we are focusing on a generic point of the desired ambit but we can trivially generalize the following calculations to all the points in the desired ambit. Let us calculate the following partial derivatives
	\begin{equation}
\begin{cases}
	\frac{\partial \epsilon_1}{\partial e} &= -Q(e, \tilde e) - \frac{\partial Q(e, \tilde e)}{\partial e}e\\
	\frac{\partial \epsilon_1}{\partial \tilde e} &= -\frac{\partial Q(e,\tilde e)}{\partial \tilde e}e \\
	\frac{\partial \epsilon_2}{\partial e} &= -F(e, \tilde e) - \frac{\partial F(e, \tilde e)}{\partial e}e\\
	\frac{\partial \epsilon_2}{\partial \tilde e} &= -\frac{\partial F(e, \tilde e)}{\partial \tilde e}e
\end{cases}.
	\end{equation}
We evaluate these partial derivatives at $\begin{bmatrix}e & \tilde e\end{bmatrix} = 0$, where $\frac{\partial Q(e,\tilde e)}{\partial \tilde e}e \big |_{e,\tilde e = 0} = \frac{\partial Q(e,\tilde e)}{\partial e}e \big |_{e,\tilde e = 0} = \frac{\partial F(e,\tilde e)}{\partial e}e \big |_{e,\tilde e = 0} = \frac{\partial F(e,\tilde e)}{\partial \tilde e}e \big |_{e,\tilde e = 0} = 0$. Note that the terms in $Q$ and $F$ are scalar products between the elements of $z$ and $\tilde z$, and their partial derivatives with respect to $e$ and $\tilde e$ do not cancel out the second term $e$ in, for example, $\frac{\partial Q(e,\tilde e)}{\partial \tilde e}e$ \cite{MarCaoJa15,MouMorseBelSunAnd15}.
Therefore, we arrive at the following linearized system at the equilibrium $\begin{bmatrix}e^T & \tilde e^T\end{bmatrix}^T = 0$
\begin{equation}
	\begin{bmatrix}\dot e \\ \dot{\tilde e}\end{bmatrix} = -\begin{bmatrix}Q(0,0) & 0 \\ F(0,0) & 0 \end{bmatrix}\begin{bmatrix} e \\ \tilde e\end{bmatrix}, \label{eq: J}
\end{equation}
where it is clear that the number of eigenvalues of the Jacobian in (\ref{eq: J}) equal to zero is at least as the dimension of $\tilde e$. Since $\mathbb{\tilde G}_{inf}$ has at least one edge more than $\mathbb{G}_{flex}$, then at least one eigenvalue of the Jacobian in (\ref{eq: J}) is equal to zero. Finally, the selection of different $\tilde d_k$ at different points of the desired ambit of $\mathbb{G}_{flex}$ might change $Q$ and $F$, but it does not change the presence of at least one zero eigenvalue in (\ref{eq: J}).
\end{proof}

A quick inspection of (\ref{eq: J}) reveals the evident null impact of the
selection of the virtual desired distances $\tilde d_k$ on the stability of the controlled error signal $e$ since $Q$ is a positive definite matrix. The fact that the self-contained linearized system (\ref{eq: J}) contains at least one eigenvalue equal to zero compromises its stability against small disturbances. Consequently, the new shifted equilibrium in a neighborhood of the origin of the augmente error system for every point in the desired ambit might become unstable.

For example, consider the situation where agents $i$ and $j$ in the edge $\mathcal{E}_k$ have different understandings about the target distance between them\footnote{By rearranging terms, this situation can be seen as having biased distance sensors.}, denoted by $d_k^2$ and $d_k^2 + \mu_k$ respectively. Without lose of generality, consider that the mismatch is at ${\mathcal{E}_k^{\text{tail}}}$, then we can rewrite (\ref{eq: pdot}) as
\begin{equation}
\dot p = -R(z)^Te -\overline SD_z \mu,
	\label{eq: pdotmu}
\end{equation}
where $\mu\in\R^{|\mathcal{E}|}$ is the stacked vector of mismatches $\mu_k$, and $S\in\R^{n\times|\mathcal{E}|}$ is the result of taking the incidence matrix $B$ and replacing its $-1$'s entries by $0$. Together with the \emph{mismatched} error system, we can derive the dynamics of the \emph{mismatched} virtual error system by following the same steps as for arriving at (\ref{eq: etildedot})
\begin{equation}
	\begin{cases}
\dot e &= -2R(z)R(z)^Te -2R(z)\overline S D_z \mu =: \epsilon_3(e,\tilde e,\mu) \\
		\dot{\tilde e} &= -2D_{\tilde z}\overline{\tilde B}^TR(z)^Te -2D_{\tilde z}\overline{\tilde B}^T\overline S D_z \mu =: \epsilon_4(e,\tilde e,\mu),
	\end{cases}
	\label{eq: eaugmu}
\end{equation}
and by following the same steps as in the proof of Theorem \ref{thm: 1} we arrive at the linearization of the self-contained augmented error system at the perturbed equilibrium by $\mu$
\begin{equation}
	\begin{bmatrix}\dot e \\ \dot{\tilde e}\end{bmatrix} = \begin{bmatrix}\frac{\partial \epsilon_3(e, \tilde e, \mu)}{\partial e} & \frac{\partial \epsilon_3(e, \tilde e, \mu)}{\partial \tilde e} \\ \frac{\partial \epsilon_4(e, \tilde e, \mu)}{\partial e} & \frac{\partial \epsilon_4(e, \tilde e, \mu)}{\partial \tilde e}\end{bmatrix}\begin{bmatrix} e \\ \tilde e\end{bmatrix},
\label{eq: Jmu}
\end{equation}
where for a system free of mismatches or calibration errors, i.e., $\mu = 0$, the system (\ref{eq: Jmu}) equals (\ref{eq: J}). Firstly, note in (\ref{eq: eaugmu}) that only $e = 0$ does not imply $\dot {\tilde e} = 0$ anymore. In fact, the perturbed equilibrium at the origin of the self-contained augmented error system \emph{forces} some (possibly all) $\tilde d_k$ to a specific value. Consequently, the ambit of the perturbed flexible formation is reduced. Secondly, we need to carry out an analysis on the sensitivity of the zero eigenvalues of (\ref{eq: Jmu}) in order to check the stability of the perturbed equilibrium. While the perturbed equilibrium is a continuous function of $\mu$ (therefore close to the desired one for small disturbances), it might become unstable. On top of that, the perturbed formation might not converge to a point of its orbit, since the velocities of the agents in system (\ref{eq: pdotmu}) might not converge to zero. A very detailed example of the presented findings for the case of having three agents and two edges can be found in \cite{de2017controlling}.

\section{Applications}
\label{sec: ex}
\subsection{Satellites around rigid formations}
Consider $n$ agents in a minimally infinitesimally rigid formation with $\mathbb{G}_{inf}$, and add $m$ agents such that each new agent is only linked with one of the agents belonging to the rigid formation. Therefore the resulting formation is the union of a minimally infinitesimally rigid and a flexible formation. We call \emph{satellite} and \emph{rigid} agents to the ones belonging to the flexible and rigid part of the formation respectively. The goal of the algorithm in this subsection is to form an infinitesimally rigid shape with the rigid agents, while the satellite agents orbit the rigid agents, where the radius and the angular velocity of the satellite agents can be set by design. Note that by setting the radius, we are setting the ambit of the corresponding $z_k$ in the flexible formation. Also recall that all the agents cooperate in their respective edges, i.e, the underlying graph for describing the formation is undirected and not directed. For example, this fact has implications in the robustness and the convergence speed of the system. In fact, since all the agents are interconnected, \emph{a priori} is not trivial to see what is the impact in the whole formation when some agents have the goal of orbiting and controlling a desired distance at the same time.

Without lose of generality, set all the satellite agents to be the tail in their corresponding edges $\mathcal{E}_k$ with the rigid agents, and order the edges such that the ones for the satellite agents are the last ones. Let us define the matrix $S\in\R^{n\times|\mathcal{E}|}$ as the result of setting all the elements of the incidence matrix $B$ to zero excepting the terms corresponding to the tails of the edges where a satellite agent is involved. Let the $z_k^\bot$ be the $\frac{\pi}{2}$ radians clockwise rotated version of $z_k$, and stack all of them in the vector $z^\bot\in\R^{2|\mathcal{E}|}$. Finally, we propose the following extension to the standard gradient based controller in (\ref{eq: pdot})
\begin{equation}
\dot p = - R(z)^Te + \overline S z^\bot.
\label{eq: pzper}
\end{equation}
In particular, the satellite agent $i$ in the edge $\mathcal{E}_k$ follows
\begin{equation}
\dot p_i = -z_ke_k + z_k^\bot,
	\label{eq: porb}
\end{equation}
while the rigid agents just implement the standard gradient control (\ref{eq: dotpigrad}). Note that while so far we have assumed that $\mu_k = a\in\R$, we are not limited to only real numbers, so we have $\mu_k = A\in\R^{2\times 2}$ in (\ref{eq: porb}), as it has also been considered in \cite{de2017distributed}, e.g., we could consider a failure in the sensor where the two components of $z_k$ are mixed.

\begin{theorem}
Consider the system (\ref{eq: pzper}) where the formation is the result from the union of a minimally infinitesimally rigid formation with $n$ agents, and a flexible formation consisting only of $m$ satellite agents. Then, the desired ambit $\mathcal{A}z$ of the flexible formation with the set (\ref{eq: D}), which defines a desired infinitesimally rigid shape with attached flexible edges, is locally stable. Furthermore, the ambit is locally stable on a point of its orbit, i.e., the rigid agents will stop in the plane while the satellite agents will orbit around their corresponding rigid agents.
\end{theorem}
\begin{proof}
We start by deriving the dynamics of the stacked vector of relative positions $z$
\begin{equation}
\dot z = -\overline B^T R(z)^Te + \overline B^T\overline S z^\bot.
	\label{eq: zp}
\end{equation}
Note that because of how we have defined $S$ and $B$ yields
\begin{equation}
B^TS = \begin{bmatrix}0 & 0 \\ 0 &  I_m\end{bmatrix},
\end{equation}
where $I_m$ is the $m\times m$ identity matrix. We further derive the dynamics of the error system
\begin{align}
	\dot e &= -2D_z^T\overline B^T R(z)^Te + 2D_z^T\overline B^T\overline S z^\bot
	\nonumber \\
	&= -2D_z^T\overline B^T R(z)^Te + 2D_z^T \begin{bmatrix}0 & 0 \\ 0 &  \overline I_m\end{bmatrix} z^\bot \nonumber \\
		&= -2D_z^T\overline B^T R(z)^Te + 2\begin{bmatrix}0 \\ \vdots \\ z^T_{|\mathcal{E}|-m+1}z^\bot_{|\mathcal{E}|-m+1} \\ \vdots \\ z^T_{|\mathcal{E}|}z^\bot_{|\mathcal{E}|} \end{bmatrix} \nonumber \\
	&= -2R(z)R(z)^Te,
\end{align}
	where we have used the relation $z_k^Tz_k^\bot = 0$. The vanishing of the second term of the error dynamics was not totally unexpected. In fact, the second term in the dynamics of the satellite agent in (\ref{eq: porb}) is not contributing to get closer or further to the corresponding rigid agent, therefore it does not have any impact in the control of the corresponding $d_k$. For the considered formation, it is not difficult to check that $R(z)^T$ is full row rank in the neighborhood of the desired ambit. Therefore, by choosing (\ref{eq: Ve}) as Lyapunov function and invoking LaSalle's principle, we can conclude that the error signal $e(t) \to 0$ exponentially fast as $t\to\infty$ if $||e(0)||$ is sufficiently small \cite{MaJaCa15}.
 
We continue by analyzing the dynamics of $z$ in (\ref{eq: zp}), where it is clear that
	\begin{equation}
	\begin{cases}
		\dot z_k(t) \to 0, k\in\{1,\dots,|\mathcal{E}|-m\} \\
	\dot z_k(t) \to Hz_k(t), k\in\{|\mathcal{E}|-m+1, \dots, |\mathcal{E}|\} \end{cases} \ \text{as} \ t\to\infty, \nonumber
	\end{equation}
	where $H = \begin{bmatrix}0 & -1 \\ 1 & 0\end{bmatrix}$ 
    shows the steady state rotatory motion of the last $m$ relative positions. Therefore, we can conclude that $z(t) \to \mathcal{A}z$ as $t\to\infty$, i.e., the desired ambit of the formation is locally stable.

We finally check the system $(\ref{eq: pzper})$, concluding that the first $n$ agents' velocities belonging to the minimally infinitesimally rigid formation converge to zero as the error signal converges to zero, and the last $m$ agents travel in an orbit around their corresponding rigid agents. Therefore, we can also conclude that the desired ambit $\mathcal{A}z$ is also locally stable on a point of its orbit in the plane.
\end{proof}

Obviously, the angular velocity of the satellite agents can be set by multiplying $z_k^\bot$ in (\ref{eq: porb}) by a factor $\omega_k\in\R$. We show a simulation in Figure \ref{fig: orbits} of a formation with four rigid agents and four satellite agents. 

We have simulated more complex patterns like the one in Figure \ref{fig: jansen} by pinning down some of the agents. However, while we can achieve stable motions, the errors are not driven to zero since some of the agents do not compensate the steady state motion of their neighbors. The simulations indicate that the error signals follow a periodic signal that is induced from the forced constant rotational motions. This insight indicates that estimators for the compensation of harmonic disturbances based on the internal model principle can be helpful as it has been demonstrated in \cite{MarCaoJa15}.

\begin{figure}
\centering 
\includegraphics[width=0.5\columnwidth]{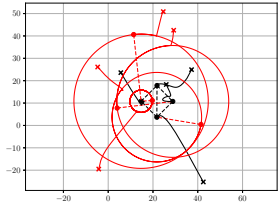} 
\caption{A flexible formation consisting of four rigid agents (black color) forming a square, with four satellite agents (red color). The crosses are the initial positions, and the dashed lines are the controlled distances. The formed square together with the circular trajectories defines the desired ambit $\mathcal{A}z$ of the flexible formation.}
\label{fig: orbits}
\end{figure}

\subsection{Converging to a specific shape in the ambit of a flexible formation}
As discussed in Section \ref{sec: aug}, the deliberated introduction of biases into the range sensors for the gradient based controller can reduce the ambit of a flexible formation. In particular, we can design such parameters in order to reduce the ambit of the new equilibrium to the empty set. Indeed, we can stabilize a shape to a point of its orbit (or to a desired orbit if we want a travelling shape), but by employing fewer edges than in the standard rigidity-based controllers, and of course, without requiring extra resources. For example, the system is still distributed and the agents employ their own  coordinate frames in performing local measurements. Finally, the stability has to be assessed by checking the eigenvalues of the Jacobian in (\ref{eq: Jmu}). Although we are investigating systematic methods in order to avoid the eigenvalue calculation, we provide the following example
\begin{equation}
	\begin{cases}
		\dot p_1 &= -z_1e_1 + z_4e_4 + z_1 - z_4^\bot\\
		\dot p_2 &= -z_2e_2 + z_1e_1 \\
		\dot p_3 &= -z_3e_3 + z_2e_2 + z_3 - z_2^\bot\\
		\dot p_4 &= -z_4e_4 + z_3e_3
	\end{cases},
	\label{eq: psq}
\end{equation}
where $B = \left (\begin{smallmatrix}
1 & 0 & 0 & -1 \\
-1 & 1 & 0 & 0 \\
0 & -1 & 1 & 0 \\
0 & 0 & -1 & 1
\end{smallmatrix}\right)$,
with $d_{\{1,2,3,4\}} = d > 0$. It is straightforward to check that $e = 0$ with $z_1 = z_4^\bot$ and $z_3 = z_2^\bot$ is a set of desired equilibra of (\ref{eq: psq}), i.e., we are restricting the flexible parallelogram to a square. More precisely, the ambit of the \emph{disturbed} flexible shape is $\mathcal{A}z = \emptyset$. The stability of the shape is checked in (\ref{eq: Jmu}) by choosing $\tilde z_1 = p_1 - p_3$ and $\tilde d_1 = \sqrt{2}d$, and details about how to proceed with the calculations can be found in \cite{de2017controlling}. We show in Figure \ref{fig: sq} a simulation of the system (\ref{eq: psq}) with $d=10$.


\begin{figure}
\centering
\includegraphics[width=0.5\columnwidth]{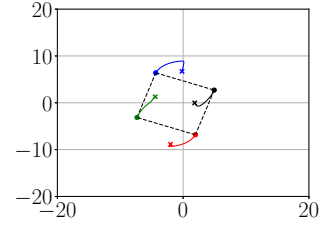}
	\caption{A simulation of the system (\ref{eq: psq}). The \emph{disturbed} flexible formation converges to a square. Only four edges are needed, while for the standard rigid distance-based control at least five would be needed. The crosses denote the initial positions.}
\label{fig: sq}
\end{figure}

\section{Conclusions}
\label{sec: con}
This paper has studied the stability of the distance-based undirected flexible formations. Unlike the stabilization control of infinitesimally rigid formations, the perturbed desired equilibrium of a flexible formation can be turned unstable regardless of the magnitude of the disturbance in the range sensors of the agents. In order to check such a stability we have constructed an augmented error system. This system is created by adding virtual edges to the flexible formation until it becomes virtually minimally rigid. We have provided examples of how to exploit these disturbances as design parameters. In particular, it is possible to control specific shapes with fewer edges than in rigid formations, and to control the motion of the agents in the flexible edges while mixing with rigid formations. The next step to take in our research will focus on finding more systematic methods for exploiting these design parameters in flexible formations. In particular, it will be an interesting research topic to find algorithms that can avoid the necessity of checking the stability of the augmented error system.

\bibliographystyle{IEEEtran}
\bibliography{hector_ref}

\end{document}